\documentclass[pdflatex,sn-basic,iicol,Numbered]{sn-jnl}

\usepackage{graphicx}%
\usepackage{amsmath,amssymb,amsfonts}%
\usepackage{amsthm, bbm}%
\usepackage{mathrsfs}%
\usepackage[title]{appendix}%
\usepackage{xcolor}%
\usepackage{textcomp}%
\usepackage{manyfoot}%
\usepackage{booktabs}%
\usepackage{algorithm}%
\usepackage{algorithmicx}%
\usepackage{algpseudocode}%
\usepackage{listings}%
\usepackage{boxedminipage}
\usepackage{multirow}%
\usepackage{makecell}
\usepackage{booktabs} 
\usepackage{array}    

\theoremstyle{thmstyleone}%
\newtheorem{theorem}{Theorem}
\theoremstyle{thmstyletwo}%
\newtheorem{remark}{Remark}%
\newtheorem{lemma}{Lemma}%

\theoremstyle{thmstylethree}%
\newtheorem{definition}{Definition}%
\newtheorem{corollary}{Corollary}

\begin{document}


\title[Efficient Error-Detecting PIR]{Efficient DPF-based Error-Detecting Information-Theoretic Private Information Retrieval Over Rings\textsuperscript{1}}


\author[1]{\fnm{Pengzhen} \sur{Ke}}\email{kepzh@shanghaitech.edu.cn}
\author*[1,3]{\fnm{Liang Feng} \sur{Zhang}}\email{zhanglf@shanghaitech.edu.cn}
\author[2]{\fnm{Huaxiong} \sur{Wang}}\email{hxwang@ntu.edu.sg}
\author[3]{\fnm{Li-Ping} \sur{Wang}}\email{wangliping@iie.ac.cn}












\abstract{
Authenticated private information retrieval (APIR) is the state-of-the-art error-detecting private information retrieval (ED-PIR), using Distributed Point Functions (DPFs) for subpolynomial complexity and privacy. However, its finite field structure restricts it to prime-order DPFs, leading to prohibitively large key sizes under information-theoretic settings, while its dual-DPF-key design introduces unnecessary communication overhead, limiting its practicality for large-scale deployments.

This paper proposes a novel ring-based information-theoretic ED-PIR (itED-PIR) scheme that overcomes these limitations by leveraging prime-power-order information-theoretic DPFs (itDPFs). Built over a prime-power ring, the proposed scheme breaks APIR’s field-induced constraint to enable more efficient DPF utilization, significantly reducing key size growth and rendering the scheme feasible for high-security scenarios. Additionally, a single-itDPF-key design halves query-side communication overhead by eliminating APIR’s redundant dual-key setup, without compromising privacy or verifiability. 

Beyond immediate efficiency gains, this work establishes a lightweight, flexible framework for constructing DPF-based malicious-resilient private information retrieval, opening new avenues for privacy-preserving data retrieval in distributed storage systems and post-quantum privacy protocols.
}




\maketitle
\footnotetext[1]{This work has been accepted for publication in Cybersecurity and is currently in press.}

\section{Introduction}
\label{sec:intro}
Private information retrieval (PIR) enables a client to retrieve a target database block \( x_{\alpha} \in {\bf x} = (x_1, \ldots, x_n) \) while keeping the index \( \alpha \in [n] \) private from servers. For large-scale data, \cite{chor1998private} proves that a trivial PIR scheme with \( O(n) \) communication complexity is impractical. 
Thus, multi-server information-theoretic PIR \citep{beimel2001information, yekhanin2008towards, efremenko20093} has become the dominant direction. In this model, the database is replicated across servers, and \( t \)-privacy guarantees no coalition of up to \( t \) servers infers the retrieval index \( \alpha \).

A critical challenge in multi-server PIR is addressing malicious servers: these may return incorrect responses (due to attacks, network failures, or outdated data), leading to erroneous retrieval results. Among existing solutions for error detection, Authenticated PIR (APIR) \citep{colombo2023authenticated} is a state-of-the-art Error-Detecting PIR (ED-PIR) scheme that leverages Distributed Point Functions (DPFs) to achieve subpolynomial communication complexity, providing a foundational solution for privacy-preserving retrieval with error detection. However, APIR suffers from following limitations. 

First, APIR relies on a dual-DPF-key design (two DPF keys per server) to implement its verification logic. This redundant design introduces unnecessary communication overhead.

Second, APIR is strictly constrained to finite field constructions, limiting its compatibility to DPFs with prime-order output groups \( \mathbb{G} = (\mathbb{Z}_p, +) \). This constraint excludes the use of more efficient prime-power-order output group DPFs (over \( \mathbb{Z}_{p^\tau} \) for \( \tau \geq 2 \)) and leads to inefficiency. 

These limitations motivate the development of a new ED-PIR scheme that retains the subpolynomial communication efficiency of APIR’s DPF-based framework, while eliminating its key redundancy and field-related constraints. To achieve these objectives, this paper presents a ring-based information-theoretic ED-PIR scheme instantiated from prime-power-order output group DPFs, which delivers significant efficiency improvements over APIR by addressing both of its foundational limitations.

\subsection{Related Work} \label{sec:related_work}
{\bf Information-Theoretic PIR.} 
Multi-server information-theoretic PIR (itPIR) was first formalized by \cite{chor1998private}, with subsequent works \cite{beimel2001information, yekhanin2008towards, efremenko20093} optimizing its communication complexity to sublinear or even polylogarithmic levels. These schemes achieve unconditional $t$-privacy but assume honest-but-curious servers, lacking resilience to malicious responses.

\vspace{2mm}
\noindent
{\bf Malicious-Resilient PIR.}
To mitigate the risks posed by malicious responses, malicious-resilient PIR schemes \citep{beimel2002robust, yang2002private, goldberg2007improving, devet2012optimally, zhang2014verifiable, sun2017capacity, sun2017capacity2, banawan2018capacity, kurosawa2019correct, zhao2021verifiable, ben2022verifiable, ke2022two, zhang2022byzantine, zhu2022post, colombo2023authenticated, ke2023private} introduce verification guarantees that allow the client to detect or even correct incorrect responses. 
Among these schemes, ED-PIR \citep{ke2022two, ke2023private, eriguchi2022optimal} achieves information-theoretic privacy and is $(\ell - 1,\epsilon)$-verifiable for $\ell$ servers with polylogathic communication complexity. In contrast, APIR \citep{colombo2023authenticated} also achieves information-theoretic privacy but is $(1, \epsilon)$-verifiable, and has subpolynomial communication complexity.

\vspace{2mm}
\noindent
{\bf Information-Theoretic DPF.}
Information-Theoretic DPF (itDPF) was first introduced by \cite{boyle2023information}, which presents a statistically private scheme for 3 servers (with \( t=1 \)) and output additive group \( \mathbb{Z}_p \) (for \( p \geq 2 \)), as well as a perfectly private scheme for 4 servers and output additive group \( \mathbb{Z}_{p^\tau} \) (for \( p \geq 3 \)). The work by \cite{li2025efficient} completes the perfectly private scheme for 4 servers with output additive group \( \mathbb{Z}_{p^\tau} \) (for \( p=2 \)); additionally, it further proposes a \( t \)-private scheme for any \( t \geq 1 \), along with a scheme where the key size is subpolynomially related to the size of the output field.

\begin{table*}[tbp]
    \centering
    \setlength{\tabcolsep}{8pt} 
    \renewcommand{\arraystretch}{1.4} 
    \begin{tabular}{@{}ccccccc@{}} 
        \toprule
        \textbf{Scheme} & \textbf{$\ell$} & \textbf{$t$} & \makecell{\textbf{Algebraic}\\ \textbf{Structure}} & \textbf{Privacy} & \textbf{Verifiable} &     \makecell{\textbf{Communication}\\ \textbf{Complexity}} \\
        \midrule
        \multirow{2}{*}{Corollary \ref{coro:1}} & 3 & 1 & $\mathbb{R}_p$ & Statistical & $\bigl(1, \frac{1}{p-1}\bigr)$ & $O\bigl(\lambda \log p \cdot 2^{c_1(p)s(n)}\bigr)$ \\
         & 4 & 1 & $\mathbb{R}_p$ & Statistical & $\bigl(1, \frac{1}{p-1}\bigr)$ & $O\bigl(\lambda \cdot 2^{10s(n)} + \lambda \log p\bigr)$ \\
        Corollary \ref{coro:2} & 4 & 1 & $\mathbb{R}_{p^{\tau}}$ & Perfect & $\bigl(1, \frac{1}{p^\tau-p^{\tau-1}}\bigr)$ & $O\bigl(\tau \log p \cdot 2^{c_2(p)s(n)}\bigr)$ \\
        Corollary \ref{coro:3} & 8 & 1 & $\mathbb{R}_p$ & Perfect & $\bigl(1, \frac{1}{p-1}\bigr)$ & $O\bigl(2^{10s(n)} + \log p\bigr)$ \\
        Corollary \ref{coro:4} & $d(t+1)$ & $t$ & $\mathbb{R}_p$ & Perfect & $\bigl(t, \frac{1}{p-1}\bigr)$ & $O\bigl(\log p \cdot n^{1/\lfloor (2d+1)/t \rfloor}\bigr)$ \\
        \midrule
        \multirow{2}{*}{$\text{APIR}^{\dagger}$} & 3 & 1 & $\mathbb{F}_p$ & Statistical & $\bigl(1, \frac{1}{p-1}\bigr)$ & $O\bigl(\lambda \log p \cdot 2^{c_1(p)s(n)}\bigr)$  \\
                                  &  4 & 1 & $\mathbb{F}_{p}$ & Perfect & $\bigl(1, \frac{1}{p-1}\bigr)$ & $O\bigl(\log p \cdot 2^{2p\cdot s(n)}\bigr)$ \\
        \bottomrule
    \end{tabular}
    \vspace{0.5em}
    \caption{Concrete instantiations information-theoretic ED-PIR Schemes ($\ell$: the number of servers; $t$: the number of colluding servers; $p$: a prime; $d,\tau$: positive integers; 
    $\mathbb{R}_p, \mathbb{R}_{p^{\tau}}$: the ring $\mathbb{R}_p = (\mathbb{Z}_p,+,\cdot), \mathbb{R}_{p^{\tau}} = (\mathbb{Z}_{p^{\tau}},+,\cdot)$;
    $c_1, c_2$: the function such that $c_1(2)=6, c(3)=10, c(p)=2p$ for $p\geq 5$), $c_2(2)=6, c_2(p)=2p$ for $p\geq 3$;  $s(n)$: shorthand for $\sqrt{\log n \log \log n}$. 
    }
    \parbox{\linewidth}{\footnotesize\raggedright $\dagger$: APIR \citep{colombo2023authenticated}}
    \label{tab:itedpir_concrete}
\end{table*}

\subsection{Contributions}
This paper addresses the key inefficiencies of the state-of-the-art APIR scheme in \cite{colombo2023authenticated} under the information-theoretic setting by proposing a novel ring-based information-theoretic ED-PIR scheme, with concrete diverse instantiations presented in Table \ref{tab:itedpir_concrete}.

As the first contribution, the information-theoretic ED-PIR scheme is constructed over the prime-power ring $\mathbb{R}=(\mathbb{Z}_{p^\tau},+,\cdot)$, breaking itAPIR’s finite field constraint that limits support to prime-order information-theoretic DPFs. This enables the design of a novel $4$-server perfectly $1$-private information-theoretic ED-PIR scheme with subpolynomial communication complexity (Corollary \ref{coro:2} in Table \ref{tab:itedpir_concrete}). Additionally, a single-DPF-key query framework is designed for ED-PIR, eliminating the dual-key redundancy of APIR and halving query-side communication overhead without any tradeoffs in privacy or verifiability.

As the second contribution, a diverse set of concrete information-theoretic ED-PIR instantiations are provided (Table \ref{tab:itedpir_concrete}), covering $3$, $4$, $8$, and general $d(t+1)$ server architectures, both statistical and perfect privacy, and collusion tolerance for $t=1$ to arbitrary $t\geq1$. 
This offers solutions for distributed storage systems with varying requirements, while maintaining subpolynomial communication complexity.

\subsection{Technique Overview}
\label{sec:our_approach}
This work proposes a novel information-theoretic ED-PIR scheme constructed over the ring $\mathbb{R}=(\mathbb{Z}_{p^\tau},+,\cdot)$, which leverages a single DPF key per server to reduce query-side communication overhead.

Error detection is realized via a ring-based verification check. Taking a binary database as an example, suppose the client intends to retrieve the element \( x_\alpha \in \{0,1\} \) from the database \( {\bf x} = (x_1, \dots, x_n) \). The client randomly selects an invertible element \( \beta \in \mathbb{R}^\ast \) (the unit group of \( \mathbb{R} \)), then sends the DPF key corresponding to the point function \( f_{\alpha, \beta} \) to each server, and receives responses from the servers. In the absence of interference from malicious servers, the sum of these responses should be \( {\sf res} = \sum_{i \in [n]} f_{\alpha, \beta}(i) \cdot x_i \), which equals \( x_\alpha \in \{0,1\} \). If the result does not belong to \( \{0,1\} \), the verification check identifies it as erroneous. The privacy guarantee of the DPF ensures that servers cannot obtain any information about \( \alpha \) or \( \beta \), thus protecting the privacy of \( \alpha \) while ensuring that the probability of malicious servers tampering with the result such that the modified \( \hat{\sf res} \in \{0,1\} \) is extremely small. This novel error detection technique enables sending only one DPF key to each server during the query phase.

The ring structure guarantees that for any DPF with a given output group, the scheme can always adopt the ring corresponding to that group as the data space. Thus, the scheme imposes no restrictions on the underlying DPF and is compatible with DPFs of any output group.

The formal construction, rigorous security analysis, and concrete instantiations of the itED-PIR scheme with explicit communication complexity bounds are presented in the subsequent sections.

\section{Preliminaries} \label{preliminaries}
    {\bf Notations.  }
    $\mathbb{N}$ denotes the set of natural numbers $\{1,2,3,\ldots\}$.
    For any function $G: \mathbb{G} \to \mathbb{H}$, let $G(\mathbb{G})$ denote the image of $\mathbb{G}$ under $G$, formally defined as the set
    $
        G(\mathbb{G}) \triangleq \{ G(s) \mid s \in \mathbb{G} \}.
    $        
    For any group $\mathbb{G}$, let $|\mathbb{G}|$ denote the order of the group.
    For any ring $\mathbb{R}$, $\mathbb{R}^*$ denotes the multiplicative group consisting of all non-zero invertible elements in $\mathbb{R}$, i.e., $\mathbb{R}^* \triangleq \{ r \in \mathbb{R} \mid r \text{ is invertible} \}$.
    \emph{Unless otherwise specified, all groups considered in this paper are abelian.}
    \vspace{2mm}
    
    \noindent
    {\bf Point Function. } Given a domain size $n$ and an group $\mathbb{G}$, the {\em point function} $f_{\alpha, \beta}: [n] \to \mathbb{G} (\alpha \in [n], \beta \in \mathbb{G})$ is a function such that for any $i\in [n]$,
    \begin{equation}        
        \label{eq:point function}
        f_{\alpha,\beta}(i) = \begin{cases} 
            \beta, & \text{if } i = \alpha \\
            0, & \text{if } i \neq \alpha
        \end{cases}.
    \end{equation}

    \subsection{Information-Theoretic Distributed Point Function} 
    \label{sec:DPF}    
    Informally, an $\ell$-server information-theoretic Distributed Point Function (itDPF) scheme involves $\ell$ servers $\{\mathcal{S}_j\}_{j \in [\ell]}$ and a client $\mathcal{C}$. The scheme enables the client to evaluate a point function $f_{\alpha, \beta}(i)$ across the $\ell$ servers for any $i, \alpha \in [n]$, an output group $\mathbb{G}$ and $\beta \in \mathbb{G}$, while preserving the privacy of $\alpha$ and $\beta$.
    
    \begin{definition}[{\bf itDPF} \citep{gilboa2014distributed,boyle2016function, boyle2023information, li2025efficient}]
        An ($t$-private) $\ell$-server itDPF scheme $\Pi$ with domain size $n$ and output group $\mathbb{G}$ is a tuple of $\ell + 1$ algorithms $\Pi = (\mathsf{Gen}, \{ \mathsf{Eval}_j \}_{j \in [\ell]})$, defined as follows:
        \begin{itemize}
            \item $\{k_j\}_{j \in [\ell]} \leftarrow \mathsf{Gen}(1^\lambda, f_{\alpha, \beta})$: A randomized key generation algorithm run by the client. It takes as input the security parameter $\lambda$ and the point function $f_{\alpha, \beta}$ (with $\alpha \in [n]$ and $\beta \in \mathbb{G}$), and outputs $\ell$ keys $\{k_j\}_{j \in [\ell]}$. Each key $k_j$ is sent to server $\mathcal{S}_j$.
        
            \item $y_j \leftarrow \mathsf{Eval}_j(k_j, i)$: A deterministic evaluation algorithm run by server $\mathcal{S}_j$. It takes as input the key $k_j$ and index $i$, and outputs a group element $y_j \in \mathbb{G}$.
        \end{itemize}
    \end{definition}

    \noindent

\begin{table*}[tbp]
    \centering
    \setlength{\tabcolsep}{6pt} 
    \renewcommand{\arraystretch}{1.4} 
    \begin{tabular}{@{}lccccc@{}} 
        \toprule
        \textbf{Scheme} & $\boldsymbol{\ell}$ & $\boldsymbol{t}$ & \textbf{Key Size} & $\boldsymbol{\mathbb{G}}$ & \textbf{Privacy} \\
        \midrule
        \makecell{\cite{boyle2023information},\\ Theorem 9} & 3 & 1 & 
        \makecell{$O\bigl(\lambda \log p \cdot 2^{c(p)\sqrt{\log n \log \log n}}\bigr)$,\\ $c(p) = \begin{cases} 6, & p=2; \\ 10, & p=3; \\ 2p, & p \geq 5. \end{cases}$} 
        & $\mathbb{Z}_p \ (p \geq 2)$ & Statistical \\
        \midrule 
        \makecell{\cite{li2025efficient},\\ Theorem 11} & 4 & 1 & $O\bigl(\lambda \cdot 2^{10\sqrt{\log n \log \log n}} + \lambda \log p\bigr)$ & $\mathbb{Z}_p \ (p \geq 2)$ & Statistical \\
        \midrule
        \makecell{\cite{boyle2023information},\\ Theorem 5} & 4 & 1 & $O\bigl(\tau \log p \cdot 2^{2p\sqrt{\log n \log \log n}}\bigr)$ & $\mathbb{Z}_{p^\tau} \ (p \geq 3)$ & Perfect \\
        \midrule
        \makecell{\cite{li2025efficient},\\ Theorem 6} & 4 & 1 & $O\bigl(\tau \log p \cdot 2^{6\sqrt{\log n \log \log n}}\bigr)$ & $\mathbb{Z}_{2^\tau}$ & Perfect \\
        \midrule
        \makecell{\cite{li2025efficient},\\ Theorem 8} & 8 & 1 & $O\bigl(2^{10\sqrt{\log n \log \log n}} + \log p\bigr)$ & $\mathbb{Z}_p \ (p \geq 2)$ & Perfect \\
        \midrule
        \makecell{\cite{li2025efficient},\\ Theorem 7} & \makecell{$d(t+1),$\\ $d \geq 1$} & $\geq 1$ & $O\bigl(\log p \cdot n^{1/\lfloor (2d+1)/t \rfloor}\bigr)$ & $\mathbb{Z}_p \ (p \geq 2)$ & Perfect \\
        \bottomrule
    \end{tabular}
    \vspace{0.5em}
    \caption{Comparison of itDPF schemes ($\ell$: the number of servers; $t$: the number of colluding servers; $\mathbb{G}$: the output additive group; $p$: a prime; $\tau$: a positive integer; $\lambda$: the statistical security parameter).}
    \label{tab:existing_itdpf}
\end{table*}

    The scheme $\Pi$ must satisfy the correctness and privacy properties.

    \begin{definition}[itDPF Correctness] \label{def:itdpf-correctness}
        Informally, the scheme is correct if the sum of outputs from all evaluation algorithms equals the point function value.

        Formally, for all $\lambda$, $f_{\alpha, \beta}$, and $i \in [n]$, if $\{k_j\}_{j \in [\ell]} \leftarrow \mathsf{Gen}(1^\lambda, f_{\alpha, \beta})$, then:
        \begin{equation}            
        \Pr\left[ \sum_{j \in [\ell]} \mathsf{Eval}_j(k_j, i) = f_{\alpha, \beta}(i) \right] = 1.
        \end{equation}
    \end{definition}

    \begin{definition}[itDPF Perfect $t$-Privacy] \label{def:itdpf-perfect-privacy}
        Informally, an itDPF scheme $\Pi = ({\sf Gen}, \{\mathsf{Eval}_j\}_{j \in [\ell]})$ is perfect $t$-private if any coalition of at most $t$ servers gains {\bf absolutely no information} about the point function $f_{\alpha, \beta}$.
    
        Formally, for any $\lambda \in \mathbb{N}$, $n \in \mathbb{N}$, abelian groups $\mathbb{G}$, $\alpha_0, \alpha_1 \in [n]$, $\beta_0, \beta_1 \in \mathbb{G}$, and $T \subseteq [\ell]$ with cardinality $\leq t$, the distributions $\mathsf{Gen}_T(1^\lambda, f_{\alpha_0, \beta_0})$ and $\mathsf{Gen}_T(1^\lambda, f_{\alpha_1, \beta_1})$ are {\bf identical} (over $\mathsf{Gen}$'s random coins), where $\mathsf{Gen}_T = \{k_j\}_{j \in T}$.
    \end{definition}
    
    \begin{definition}[itDPF Statistical $t$-Privacy] \label{def:itdpf-statistical-privacy}
        Informally, the scheme is statistically $t$-private if any coalition of at most $t$ servers gains {\bf statistically no information} about $f_{\alpha, \beta}$.
    
        Formally, it satisfies the same conditions as Definition \ref{def:itdpf-perfect-privacy}, except the distributions are {\bf statistically indistinguishable} (i.e., their statistical distance is negligible in the security parameter $\lambda$).
    \end{definition}
    \begin{remark} \label{rem:itdpf-privacy-relation}
        Perfect \( t \)-privacy (Definition \ref{def:itdpf-perfect-privacy}) is strictly stronger than statistical \( t \)-privacy (Definition \ref{def:itdpf-statistical-privacy}): identical distributions are also statistically indistinguishable.
        Consequently, any itDPF scheme satisfying perfect \( t \)-privacy is naturally statistical \( t \)-privacy.
        
        For an itDPF scheme, it suffices to satisfy either perfect \( t \)-privacy or statistical \( t \)-privacy.
    \end{remark}

    Table~\ref{tab:existing_itdpf} summarizes all existing information-theoretic DPFs. From the perspective of the output group $\mathbb{G}$, these schemes can be categorized into two classes: those with prime-order $\mathbb{G}$ and those with prime-power-order $\mathbb{G}$. 
    Notably, schemes offering statistical privacy are only available for prime-order groups, while those with perfect privacy can be realized for both prime-order and prime-power-order groups. 
    Specifically, the prime-order scheme (e.g., \cite{li2025efficient}, Theorem~7) is deployable in scenarios where $t > 1$. 
    In contrast, the prime-power-order schemes (e.g., \cite{li2025efficient}, Theorem~6; \cite{boyle2023information}, Theorem~5) can be deployed with fewer servers and feature a smaller key size.

\subsection{Information-Theoretic Error-Detecting Private Information Retrieval}
Informally, an $\ell$-server information-theoretic error-detecting Private Information Retrieval (itED-PIR) scheme involves $\ell$ servers $\{\mathcal{S}_j\}_{j \in [\ell]}$ and a client $\mathcal{C}$. Each server stores a copy of the identical database $\mathbf{x} = (x_1, \ldots, x_n)$. The client $\mathcal{C}$ intends to retrieve the block $x_{\alpha}$ for some index $\alpha$. 
The scheme guarantees that the client can correctly recover $x_{\alpha}$ when all $\ell$ servers $\{\mathcal{S}_j\}_{j \in [\ell]}$ are honest and respond correctly, or the client outputs a special symbol $\perp$ to indicate incorrect responses.


\begin{figure*}[h!]
    \centering
    \begin{boxedminipage}{12cm}
        \textbf{Experiment $\mathsf{EXP}^{\rm Ver}_{\mathcal{A}, \Gamma}(n, {\bf x}, \alpha, V)$:}
        \begin{enumerate}
            \item The challenger generates $(\{q_j\}_{j \in [\ell]}, {\sf aux}) \leftarrow \mathsf{Que}(1^\lambda, n, \alpha)$ and sends the subset $\{q_j\}_{j \in V}$ to the adversary $\mathcal{A}$.
            
            \item $\mathcal{A}$ returns fraudulent responses $\{\hat{a}_j\}_{j \in V}$ to the challenger.
            
            \item The challenger computes honest responses $a_j \leftarrow \mathsf{Ans}({\bf x}, q_j, j)$ for all $j \in [\ell] \setminus V$.
            
            \item The challenger executes $\hat{\sf res} \leftarrow \mathsf{Rec}(\{ \{\hat{a}_j\}_{j \in V} \cup \{a_j\}_{j \in [\ell] \setminus V}, {\sf aux})$.
            
            \item The experiment outputs $1$ if $\hat{\sf res} \notin \{x_{\alpha}, \perp\}$ (adversary succeeds) and $0$ otherwise.
        \end{enumerate}
    \end{boxedminipage}
    \caption{The verification experiment $\mathsf{EXP}^{\rm Ver}_{\mathcal{A}, \Gamma}(n, {\bf x}, \alpha, V)$.}
    \label{fig:security experiment}
\end{figure*}

\begin{definition}[\bf itED-PIR \citep{ke2022two,ke2023private,li2025efficient,colombo2023authenticated}]
\label{def:itED-PIR}
A ($t$-private) $\ell$-server itED-PIR scheme $\Gamma = ({\sf Que}, {\sf Ans}, {\sf Rec})$ consists of three algorithms:
\begin{itemize}
    \item $(\{q_j\}_{j \in [\ell]}, {\sf aux}) \leftarrow {\sf Que}(1^{\lambda}, n, \alpha)$:  
    A randomized \emph{querying algorithm} executed by the client $\mathcal{C}$. It takes the statistical security parameter $\lambda$, database size $n$, and retrieval index $\alpha \in [n]$ as input, and outputs $\ell$ queries $\{q_j\}_{j \in [\ell]}$ (with $q_j$ sent to server $\mathcal{S}_j$) and auxiliary information ${\sf aux}$ for reconstruction.

    \item $a_j \leftarrow {\sf Ans}({\bf x}, q_j, j)$:  
    A deterministic \emph{answering algorithm} executed by server $\mathcal{S}_j$ ($j \in [\ell]$). It takes the database ${\bf x} = (x_1, \ldots, x_n)$, query $q_j$ and $j$ as input, and outputs a response $a_j$.

    \item ${\sf res} \leftarrow {\sf Rec}(\{a_j\}_{j \in [\ell]}, {\sf aux})$:  
    A deterministic \emph{reconstructing algorithm} executed by the client $\mathcal{C}$. It takes the responses $\{a_j\}_{j \in [\ell]}$ and auxiliary information ${\sf aux}$ as input, and outputs either the correct value ${\sf res} = x_{\alpha}$ when all responses are valid or the symbol ${\sf res} = \perp$ indicating at least one invalid response.
\end{itemize}
\end{definition}
\begin{remark}
    Information-theoretic security for the itED-PIR framework is categorized into two distinct notions: perfect security and statistical security. The statistical security parameter $\lambda$ is only relevant to the statistical security setting. In contrast, the perfect security setting requires no such security parameter.
\end{remark}

The scheme $\Gamma$ must satisfy the following requirements:

\begin{definition}[itED-PIR Correctness] \label{def:edpir-correctness}
    Informally, the scheme is correct if the reconstruction algorithm $\mathsf{Rec}$ outputs the correct database entry $x_\alpha$ when all $\ell$ servers respond honestly with valid answers.

    Formally, for all security parameter $\lambda \in \mathbb{N}$, database size $n \in \mathbb{N}$, database ${\bf x} = (x_1 ,...,x_n)$, and retrieval index $\alpha \in [n]$, if $(\{q_j\}_{j \in [\ell]}, \mathsf{aux}) \leftarrow \mathsf{Que}(1^\lambda, n, \alpha)$, then:
    \begin{equation}        
        \Pr\left[ \mathsf{Rec}\bigl(\{ \mathsf{Ans}({\bf x}, q_j,j)\}_{j \in [\ell]}, \mathsf{aux}\bigr) = x_\alpha \right] = 1.
    \end{equation}
\end{definition}

\begin{definition}[itED-PIR Perfect $t$-Privacy] \label{def:itedpir-perfect-tprivacy}
    Informally, an itED-PIR scheme $\Gamma$ is perfect $t$-private if any coalition of at most $t$ servers gains {\bf absolutely no information} about the client's retrieval index $\alpha$.

    Formally, for any security parameter $\lambda \in \mathbb{N}$, database size $n \in \mathbb{N}$, retrieval indices $\alpha_0, \alpha_1 \in [n]$, and subset $T \subseteq [\ell]$ with cardinality $\leq t$, the distributions $\mathsf{Que}_T(1^\lambda, n, \alpha_0)$ and $\mathsf{Que}_T(1^\lambda, n, \alpha_1)$ are {\bf identical} (over $\mathsf{Que}$'s random coins), where $\mathsf{Que}_T$ denotes the projection of $\mathsf{Que}$'s outputs to $T$, i.e., $\{q_j\}_{j \in T}$.
\end{definition}

\begin{definition}[itED-PIR Statistical $t$-Privacy] \label{def:itedpir-statistical-tprivacy}
    Informally, an itED-PIR scheme $\Gamma$ is statistically $t$-private if any coalition of at most $t$ servers gains {\bf statistically no information} about the client's retrieval index $\alpha$.

    Formally, it satisfies the same conditions as Definition \ref{def:itedpir-perfect-tprivacy}, except the distributions $\mathsf{Que}_T(1^\lambda, n, \alpha_0)$ and $\mathsf{Que}_T(1^\lambda, n, \alpha_1)$ are {\bf statistically indistinguishable} (i.e., their statistical distance is negligible in the security parameter $\lambda$).
\end{definition}

\begin{remark} \label{rem:itedpir-privacy-relation}
    Similar to the privacy relation for itDPFs (Remark \ref{rem:itdpf-privacy-relation}), perfect \( t \)-privacy (Definition \ref{def:itedpir-perfect-tprivacy}) is strictly stronger than statistical \( t \)-privacy (Definition \ref{def:itedpir-statistical-tprivacy}) for itED-PIR schemes. Any itED-PIR scheme satisfying perfect \( t \)-privacy naturally achieves statistical \( t \)-privacy.
    For an itED-PIR scheme, it suffices to satisfy either perfect \( t \)-privacy or statistical \( t \)-privacy.
\end{remark}

\begin{definition}[itED-PIR $(v, \epsilon)$-Verifiability] \label{def:edpir-verifiability}
    Informally, an itED-PIR scheme $\Gamma$ is $(v, \epsilon)$-verifiable if no coalition of up to $v$ servers can deceive the client into outputting a result $\mathsf{res} \notin \{x_{\alpha}, \perp\}$ by providing fraudulent responses for the retrieval index $\alpha$.

    Formally, for any adversary $\mathcal{A}$, any subset $V \subseteq [\ell]$ with cardinality $\leq v$, any database size $n \in \mathbb{N}$, any database ${\bf x} = (x_1,...,x_n)$, and any retrieval index $\alpha \in [n]$, the probability that the verification experiment outputs $1$ is bounded by $\epsilon$:
    \begin{equation}        
        \Pr\left[ \mathsf{EXP}^{\rm Ver}_{\mathcal{A}, \Gamma}(n, {\bf x}, \alpha, V) = 1 \right] \leq \epsilon,
    \end{equation}
    where the experiment $\mathsf{EXP}^{\rm Ver}_{\mathcal{A}, \Gamma}$ is defined in {\bf Fig.}\ref{fig:security experiment}.
\end{definition}

\begin{figure*}[h] 
    \begin{boxedminipage}{\textwidth}
        \underline{Notations}
        \begin{itemize}
            \item $\ell$: The total number of servers.
            \item $t$: The number of servers that may collude to learn the retrieval index.
            \item ${\bf x} = (x_1, \ldots, x_n)$: The binary database, which is a vector in $\{0,1\}^n$.
            \item $\alpha$: The client's retrieval index.
            \item $\mathbb{F}$: The finite field $\mathbb{F} = (\mathbb{Z}_p, +, \cdot)$ (where $p$ is a prime).
            \item $\mathbb{F}^\ast$: The multiplicative group of $\mathbb{F}$, consisting of all non-zero invertible elements (i.e., $\mathbb{F}^\ast = \mathbb{F} \setminus \{0\}$).
            \item $\mathbb{G}$: The additive group of the field $\mathbb{F}$, $\mathbb{G} = (\mathbb{Z}_{p}, +)$.
            \item $\Pi = ({\sf Gen}, \{ {\sf Eval}_j \}_{j \in [\ell]})$: A $t$-private $\ell$-server itDPF scheme with output group $\mathbb{G}$.
        \end{itemize}
        
        \underline{$( \{q_j\}_{j \in [\ell]}, {\sf aux} ) \leftarrow {\sf Que}(1^\lambda, n, \alpha)$}
        \begin{itemize}
            \item[1.] Uniformly sample a nonzero element $\beta \in \mathbb{F}^*$. Generate itDPF keys for two point functions via:
            $\{k_{j,1}\}_{j \in [\ell]} \leftarrow {\sf Gen}(1^\lambda, f_{\alpha, 1})$,
            $\{k_{j,2}\}_{j \in [\ell]} \leftarrow {\sf Gen}(1^\lambda, f_{\alpha, \beta})$.
            For each $j \in [\ell]$, construct query $q_j = (k_{j,1}, k_{j,2})$.
            \item[2.] Output queries $\{q_j\}_{j \in [\ell]}$ and auxiliary information ${\sf aux} = \beta$.
        \end{itemize}
        
        \underline{$a_j \leftarrow {\sf Ans}({\bf x}, q_j, j)$}
        \begin{itemize}
            \item[1.] Parse $q_j = (k_{j,1}, k_{j,2})$. Compute two response components:
            $a_{j,1} = \sum_{i \in [n]} x_i \cdot {\sf Eval}_j(k_{j,1}, i)$,
            $a_{j,2} = \sum_{i \in [n]} x_i \cdot {\sf Eval}_j(k_{j,2}, i)$.
            Output $a_j = (a_{j,1}, a_{j,2})$.
        \end{itemize}

        \underline{${\sf res} \leftarrow {\sf Rec}(\{a_j\}_{j \in [\ell]}, {\sf aux})$}      
        \begin{itemize}      
            \item[1.] Parse each response $a_j = (a_{j,1}, a_{j,2})$, then compute the aggregate results:
                $
                    R_1 = \sum_{j \in [\ell]} a_{j,1},  R_2 = \sum_{j \in [\ell]} a_{j,2}.   
                $
            \item[2.] Extract $\beta$ from ${\sf aux}$. If $\beta \cdot R_1 = R_2$, output $R_1$; otherwise, output $\perp$.
        \end{itemize}
    \end{boxedminipage}
    \caption{$\ell$-Server APIR Scheme Over ED-PIR Model with DPF Output Group $\mathbb{G} = (\mathbb{Z}_p$,+)}
    \label{fig:APIR}
\end{figure*}

The efficiency of a $t$-private $\ell$-server itED-PIR scheme is primarily measured by its communication complexity, defined as the total number of bits exchanged between the client and all servers.
\begin{definition}[\bf Communication Complexity] \label{communication complexity}
The \emph{communication complexity} of scheme $\Gamma$, denoted ${\sf CC}_{\Gamma}(n)$, is the maximum total number of bits communicated between the client and servers over all possible databases and retrieval indices. Formally:
\begin{equation}    
    {\sf CC}_{\Gamma}(n) = \max_{\substack{{\bf x} =(x_1,...,x_n) \\ \alpha \in [n]}}  \sum_{j=1}^{\ell} \left( |q_j| + |a_j| \right) ,
\end{equation}
where $|q_j|, |a_j|$ denote the bit-length of query $q_j$ and response $a_j$.
\end{definition}

\subsection{Information-Theoretic Authenticated Private Information Retrieval} \label{sec:APIR}
Information-theoretic authenticated Private Information Retrieval (itAPIR) \citep{colombo2023authenticated} is the most efficient multi-server itED-PIR scheme to date. For completeness, this section briefly describes a multi-server itAPIR scheme instantiated from $t$-private $\ell$-server itDPFs with the additive output group $\mathbb{G}=(\mathbb{Z}_p, +)$, as shown in Fig. \ref{fig:APIR}.

This itAPIR scheme adopts the standard itED-PIR algorithmic framework — a triple of algorithms $\Gamma = (\mathsf{Que}, \mathsf{Ans}, \mathsf{Rec})$ as defined in Definition \ref{def:itED-PIR}. The randomized query algorithm $\mathsf{Que}$ first samples a uniform non-zero element $\beta \in \mathbb{F}^\ast$, then invokes the itDPF key generation algorithm $\mathsf{Gen}$ to generate two key sets: $\{k_{j,1}\}_{j\in [\ell]} \leftarrow \mathsf{Gen}(1^\lambda, f_{\alpha, 1})$ and $\{k_{j,2}\}_{j \in [\ell]} \leftarrow \mathsf{Gen}(1^\lambda, f_{\alpha, \beta})$, corresponding to the point functions $f_{\alpha,1}$ and $f_{\alpha,\beta}$ respectively. For each server $j \in [\ell]$, the query $q_j = (k_{j,1}, k_{j,2})$ is formed by concatenating these two keys, and the algorithm outputs the query set $\{q_j\}_{j\in [\ell]}$ along with the auxiliary information $\mathsf{aux} = \beta$ for reconstruction and verification.

The deterministic answering algorithm $\mathsf{Ans}$ for server $\mathcal{S}_j$ operates as follows: upon receiving the query $q_j$ and taking the database ${\bf x}$ as input, the server parses $q_j = (k_{j,1}, k_{j,2})$ and computes two response components by evaluating the itDPF evaluation algorithm $\mathsf{Eval}_j$ over all database indices. The components are derived from weighted sums of the database elements $x_i$, with weights given by the itDPF evaluation results, i.e.,
$
    a_{j,1} = \sum_{i \in [n]} x_i \cdot \mathsf{Eval}_j(k_{j,1}, i), \quad  
    a_{j,2} = \sum_{i \in [n]} x_i \cdot \mathsf{Eval}_j(k_{j,2}, i).
$
The server then outputs the response $a_j = (a_{j,1}, a_{j,2})$ for the client's reconstruction.

The deterministic reconstructing algorithm $\mathsf{Rec}$ takes the full set of server responses $\{a_j\}_{j\in [\ell]}$ and the auxiliary information $\mathsf{aux}$ as input, first aggregating the component-wise responses across all servers to obtain two aggregated values:
$R_1 = \sum_{j\in [\ell]} a_{j,1}, \quad R_2 = \sum_{j\in [\ell]} a_{j,2}.$
It then conducts a verification check by verifying the equality $\beta \cdot R_1 = R_2$. If the equality holds, the algorithm outputs $R_1$ as the correct retrieved database entry $x_\alpha$; otherwise, it outputs the error symbol $\perp$ to indicate the presence of fraudulent or tampered server responses.

By the $t$-privacy guarantee of the underlying itDPF scheme, any adversary controlling at most $t$ colluding servers gains no information about the point functions $f_{\alpha,1}$ and $f_{\alpha,\beta}$, and thus cannot recover the secret value $\beta$ or the retrieval index $\alpha$. As a result, such an adversary is unable to construct forged response components $a_{j,1}',$ $a_{j,2}'$ for the colluding servers that would lead the client to accept a false result $R_1' \neq R_1$ while satisfying the verification equality $\beta \cdot R_1' = R_2'$.

\begin{figure*}[tbp] 
	\begin{boxedminipage}{\textwidth}
        \underline{Notations}
        \begin{itemize}
            \item $\ell$: The total number of servers.
            \item $t$: The number of servers that may collude to learn the retrieval index. 
            \item ${\bf x} = (x_1,\ldots,x_n)$: The binary database (1-bit per entry).
            \item $\alpha$: The client's retrieval index.
            \item $\mathbb{R}$: The ring $\mathbb{R} = (\mathbb{Z}_{p^\tau},+,\cdot)$ (where $p$ is a prime and $\tau$ is a positive integer).
            \item $\mathbb{G}$: The additive group of the ring $\mathbb{R}$, $\mathbb{G} = (\mathbb{Z}_{p^{\tau}}, +)$.
            \item $\Pi = (\mathsf{Gen}, \{ \mathsf{Eval}_j\}_{j\in [\ell]})$: A $t$-private $\ell$-server itDPF scheme with output group $\mathbb{G}$.
        \end{itemize}
                
        \underline{$ ( \{q_j\}_{j\in [\ell]}, \mathsf{aux} ) \leftarrow \mathsf{Que}(1^{\lambda}, n,\alpha)$}
        \begin{itemize}
		\item[1.]
            Uniformly sample an element $\beta \in \mathbb{R}^*$. Generate DPF keys via $\{k_j\}_{j \in [\ell]} \leftarrow \mathsf{Gen}(1^\lambda, f_{\alpha, \beta})$, where $f_{\alpha, \beta}$ is the point function over $\mathbb{G}$. For each $j \in [\ell]$, set query $q_j = k_j$.
            
            \item[2.] 
            Output queries $\{q_j\}_{j \in [\ell]}$ and auxiliary information $\mathsf{aux} = \beta$.
        \end{itemize}
        
        \underline{$ a_j \leftarrow \mathsf{Ans}({\bf x}, q_j,j)$}
        \begin{itemize}
                \item[1.]
                Evaluate the answer
                $
                a_j = \sum_{i \in [n]} x_{i} \cdot  \mathsf{Eval}_j (q_j,i),
                $
			and output $a_j$.
        \end{itemize}

        \underline{$\mathsf{res} \leftarrow \mathsf{Rec}(\{a_j\}_{j \in [\ell]}, \mathsf{aux})$}      
        \begin{itemize}      
            \item[1.] 
            Compute the aggregate result	
            $
            R_{\text{agg}}=\sum_{j\in [\ell]} a_{j}.
            $
            \item[2.]
            Extract $\beta$ from $\mathsf{aux}$. If $\beta^{-1} \cdot R_{\text{agg}} \in \{0,1\}$, output $\beta^{-1} \cdot R_{\text{agg}}$; otherwise output $\perp$ to indicate tampered responses.
		\end{itemize}
	\end{boxedminipage}
	\caption{$\ell$-Server itED-PIR scheme $\Gamma = (\mathsf{Que},\mathsf{Ans},\mathsf{Rec})$ for $t$-private itDPF $\Pi$ with output group $\mathbb{G} =  (\mathbb{Z}_{p^{\tau}},+)$. This base construction targets binary databases ($1$-bit data entry) and is generalized to databases with $m$-bit data entry for arbitrary $m$ in Appendix \ref{sec:mbit}.}
	\label{fig:construction}
\end{figure*}

\section{Efficient itED-PIR Scheme}\label{sec:pirrv}
This section presents a \( t \)-private \( \ell \)-server itED-PIR scheme \( \Gamma = (\mathsf{Que},\mathsf{Ans},\mathsf{Rec}) \) targeting binary databases \( {\bf x}\in \{0,1\}^n \) (see Fig. \ref{fig:construction}). 
Notably, the design is generalizable to databases with \( m \)-bit data entry for arbitrary positive integer \( m \), with the detailed construction provided in Appendix \ref{sec:mbit}.

\subsection{The Construction}  
The scheme depicted in Fig. \ref{fig:construction} is built upon a \( t \)-private \( \ell \)-server itDPF scheme \( \Pi = (\mathsf{Gen}, \{ \mathsf{Eval}_j \}_{j \in [\ell]}) \) with output group $\mathbb{G} =  (\mathbb{Z}_{p^{\tau}},+)$, where \( \mathbb{G} \) is the additive group of the ring \( \mathbb{R} = (\mathbb{Z}_{p^\tau}, +, \cdot) \). 
The core idea is that the \( t \)-privacy of the underlying itDPF protects the random invertible element \( \beta \in \mathbb{R}^* \): for any adversarial offset \( \Delta \) introduced by malicious servers, it is extremely unlikely that \( \beta^{-1} \cdot (R_{\text{agg}}+\Delta) \) fall within the valid entry set ({\{0,1\}} for binary databases, thus enabling efficient response verification via basic ring operations.

Formally, the ring \( \mathbb{R} = (\mathbb{Z}_{p^\tau}, +, \cdot) \) is defined with addition and multiplication modulo \( p^\tau \), and \( \mathbb{R}^* \) denotes the multiplicative group of \( \mathbb{R} \). 
The \( t \)-private \( \ell \)-server itED-PIR scheme \( \Gamma = (\mathsf{Que}, \mathsf{Ans}, \mathsf{Rec}) \) inherits the security of the itDPF scheme \( \Pi \) and comprises three algorithms as follows.

The querying algorithm \( \mathsf{Que} \) is executed by the client. It first uniformly samples an element \( \beta \in \mathbb{R}^* \), then invokes the itDPF key generation algorithm \( \mathsf{Gen} \) to generate a set of keys \( \{k_j\}_{j \in [\ell]} \) for point function \( f_{\alpha, \beta} \) for the target retrieval index \( \alpha \). 
Each key \( k_j \) is sent to server \( \mathcal{S}_j \) as a query, and the client retains the auxiliary information \( \mathsf{aux} = \beta \) for subsequent verification.

The answering algorithm \( \mathsf{Ans} \) is executed by each server \( \mathcal{S}_j \). Upon receiving query \( k_j \), the server evaluates the itDPF at every database index \( i \in [n] \) via \( \mathsf{Eval}_j(k_j, i) \), computes the product with the database entry \( x_i \), and sums these products over all indices to generate a response \( a_j = \sum_{i \in [n]} x_i \cdot \mathsf{Eval}_j(k_j, i) \), which is returned to the client. 

The reconstructing algorithm \( \mathsf{Rec} \) is executed by the client. It first aggregates all server responses to compute the total sum \( R_{\text{agg}} = \sum_{j \in [\ell]} a_j \), then calculates \( \beta^{-1} \cdot R_{\text{agg}} \) using the auxiliary information \( \mathsf{aux} = \beta \). 
The client outputs \( \beta^{-1} \cdot R_{\text{agg}} \) if the result lies in \( \{0,1\} \), and outputs \( \perp \) otherwise to indicate detected tampering or invalid responses. 

When all servers are honest, the final output satisfies \( \beta^{-1} \cdot R_{\text{agg}} = x_{\alpha} \in \{0,1\} \). 
When there are at most \( t \) malicious servers, since the itDPF scheme \( \Pi \) is \( t \)-private, an adversary controlling these malicious servers cannot obtain information about the point function \( f_{\alpha, \beta} \). As a result, the output \( \beta^{-1} \cdot R_{\text{agg}} \) will, with high probability, lies outside \( \{0,1\} \). This property enables verification of the correctness of the result.

\subsection{Analysis}
    This section comprehensively analyzes the properties of the itED-PIR scheme $\Gamma$ based on the itDPF scheme $\Pi$, including correctness, privacy, verifiability, and communication complexity. 
    
    \vspace{2mm}
    \noindent
    \begin{theorem} \label{thm:correctness}
        The $\ell$-server itED-PIR scheme $\Gamma$ in {\bf Fig.} \ref{fig:construction} constructed by itDPF shceme $\Pi$ is correct (Definition \ref{def:edpir-correctness}) if the scheme $\Pi$ is correct (Definition \ref{def:itdpf-correctness}).
    \end{theorem}
\begin{proof}
    For any $\lambda\in \mathbb{N}$, any $n\in \mathbb{N} $, any binary database ${\bf x}\in \{0,1\}^n$, any $(\{q_j\}_{j\in [\ell]}, {\sf aux}) \leftarrow {\sf Que}(1^{\lambda},n,\alpha)$, if ${\sf res} \leftarrow {\sf Rec}(\{{\sf Ans}({\bf x},q_j,j)\}_{j\in [\ell]},{\sf aux} )$, 
    by the correctness of the DPF scheme $\Pi$, it follows that:
    \begin{equation}        
        \sum_{j\in [\ell]} {\sf Eval}_j(q_j, i) = f_{\alpha, \beta}(i) \quad \forall i \in [n],
    \end{equation}
    where $\beta \in \mathbb{R}^*$ is the randomly sampled element retained in ${\sf aux}$. Further, the aggregated response $R_{\text{agg}}$ satisfies:
    \begin{equation}        
    \begin{aligned}            
        R_{\text{agg}} 
        &= 
        \sum_{j\in [\ell]} {\sf Ans}({\bf x},q_j,j) 
        \\
        &= 
        \sum_{j\in [\ell]}  
        \sum_{i \in [n]} 
        \left(x_{i} \cdot {\sf Eval}_j (q_j,i)\right)
        \\
        &=
        \sum_{i\in [n]} \left( x_i \cdot \sum_{j \in [\ell]} {\sf Eval}_j (q_j,i) \right)
        \\
        &=
        \sum_{i\in [n]} \left( x_i \cdot f_{\alpha, \beta}(i) \right)
        \\
        &=
        \beta \cdot x_\alpha
    \end{aligned}
    \end{equation}
    (Note: $f_{\alpha, \beta}(i) = \beta$ if $i=\alpha$, and $f_{\alpha, \beta}(i) = 0$ otherwise.) Since $\beta \in \mathbb{R}^*$ is multiplicatively invertible, it follows that:
    $
        {\sf res} = \beta^{-1} \cdot R_{\text{agg}} = x_{\alpha}, \quad \Pr \left[ {\sf res} = x_{\alpha} \right] = 1.
    $
\end{proof}

\begin{theorem} \label{thm:privacy}
    The $\ell$-server itED-PIR scheme $\Gamma$ in {\bf Fig.} \ref{fig:construction} constructed by itDPF shceme $\Pi$ is perfect $t$-private (Definition \ref{def:itedpir-perfect-tprivacy}) if the scheme $\Pi$ is perfect $t$-private (Definition \ref{def:itdpf-perfect-privacy}). 
    Similarly, $\Gamma$ is statistical $t$-private (Definition \ref{def:itedpir-statistical-tprivacy}) if $\Pi$ is statistical $t$-private (Definition \ref{def:itdpf-statistical-privacy}).
\end{theorem}

\begin{proof}
    The proof for perfect $t$-privacy of $\Gamma$ is presented as follows; the proof for statistical $t$-privacy holds identically with the term ``identical" replaced by ``statistically indistinguishable" in all subsequent steps.

    By the definition of perfect $t$-privacy for the itDPF scheme $\Pi$, for any security parameter $\lambda \in \mathbb{N}$, database size $n \in \mathbb{N}$, abelian output group $\mathbb{G}$, retrieval indices $\alpha_0, \alpha_1 \in [n]$, elements $\beta_0, \beta_1 \in \mathbb{G}$, and any subset $T \subseteq [\ell]$ with $|T| \leq t$, the distributions of $\mathsf{Gen}_T(1^\lambda, f_{\alpha_0, \beta_0})$ and $\mathsf{Gen}_T(1^\lambda, f_{\alpha_1, \beta_1})$ are {\bf identical}. Here, $\mathsf{Gen}_T$ denotes the projection of the itDPF key generation algorithm $\mathsf{Gen}$ onto the subset $T$, i.e., $\mathsf{Gen}_T = \{k_j\}_{j \in T}$.

    For the itED-PIR scheme $\Gamma = (\mathsf{Que}, \mathsf{Ans}, \mathsf{Rec})$, the view of colluding servers in $T$ is exactly the itDPF keys $\{k_j\}_{j \in T}$. 
    Therefore the distribution of the queries $\{ {\sf Que}_{T}(1^{\lambda} ,n , \alpha_0) \}$ and $\{ {\sf Que}_{T}(1^{\lambda}, n, \alpha_1)  \}$ are identical for any $\beta_0, \beta_1$, implying $\Gamma$ is perfect $t$-private.
\end{proof}

\noindent

To prove verifiability, Lemma \ref{Lemma:verifiability} characterizes the condition that an adversary must satisfy when deviating from the prescribed protocol—specifically, the tampering required for the \({\sf Rec}\) algorithm to output a result \({\sf res} \neq \perp\) (i.e., pass the security check without returning an error).

\begin{lemma} \label{Lemma:verifiability}
Let \(\Gamma\) be the \(\ell\)-server \(t\)-private itED-PIR scheme introduced in {\bf Fig.} \ref{fig:construction}. For any database size \(n \in \mathbb{N}\), any binary database \({\bf x} = (x_1, \ldots, x_n) \in \{0,1\}^n\), any index \(\alpha \in [n]\), any subset \(T \subseteq [\ell]\) with \(|T| \leq t\), and any non-zero offsets \(\{\Delta_j\}_{j \in [\ell]}\) (where \(\Delta_j = 0\) for \(j \notin T\) and \(\Delta = \sum_{j \in [\ell]} \Delta_j \neq 0\)), the following bound holds:
\begin{equation}
\Pr\left[ 
\begin{aligned}    
y \in \{0,1\} \setminus \{x_\alpha\} :&  \\
(\{q_j\}_{j \in [\ell]}, {\sf aux}) \leftarrow& {\sf Que}(1^\lambda, n, \alpha) \\
a_j \leftarrow& {\sf Ans}({\bf x}, q_j, j) \quad \forall j \in [\ell] \\
a'_j =& a_j + \Delta_j \quad \forall j \in [\ell] \\
y \leftarrow& {\sf Rec}({\sf aux}, \{a'_j\}_{j \in [\ell]})
\end{aligned}
\right]
\leq \frac{1}{|\mathbb{R}^*|},
\end{equation}
where \(\mathbb{R} = (\mathbb{Z}_{p^\tau}, +, \cdot)\) is the underlying ring, \(\mathbb{R}^*\) denotes its multiplicative group (with order \(|\mathbb{R}^*| = p^\tau - p^{\tau-1}\) by Euler's totient function), and the probability is taken over all random coins used by the algorithms.
\end{lemma}

\begin{proof}
Let \(\beta\) be the element uniformly sampled from \(\mathbb{R}^*\), \(R_{\text{agg}} = \sum_{j \in [\ell]} a_j\) denote the honest aggregate response, and \(\Delta = \sum_{j \in [\ell]} \Delta_j \neq 0\) denote the aggregate tampering offset. By the correctness of \(\Gamma\) (Theorem \ref{thm:correctness}), it holds that:
\begin{equation}    
\beta^{-1} \cdot R_{\text{agg}} = x_\alpha \in \{0,1\}.
\end{equation}

A ``false acceptance" (i.e., the tampered response is incorrectly judged as valid) occurs if and only if the reconstructed result lies in the verification set \(\{0,1\}\) but is not the true value \(x_\alpha\), which is equivalent to:
\begin{equation}
\beta^{-1} \cdot \left(R_{\text{agg}} + \Delta \right) = 1 - x_\alpha.
\end{equation}

Substituting \(\beta^{-1} R_{\text{agg}} = x_\alpha\) into the above equation, the probability stated in the lemma can be rewritten as:
\begin{equation}    
\begin{aligned}
    v &= \Pr\left[ \beta^{-1}\cdot\Delta = (1-2 x_{\alpha})\right] \\
    &= \Pr[(1-2x_{\alpha}) \beta - \Delta = 0]
\end{aligned}
\end{equation}

Since \(x_{\alpha} \in \{0,1\}\), it holds that \(1-2x_{\alpha} \in \{-1,1\}\) — both elements are non-zero and multiplicatively invertible in \(\mathbb{R}^*\). Combined with \(\Delta \neq 0\), the expression \((1-2x_{\alpha}) \beta - \Delta\) forms a non-zero linear (degree-1) constraint on \(\beta\). 

For the multiplicative group \(\mathbb{R}^*\), the non-zero linear constraint has at most one solution. Since \(\beta\) is uniformly sampled from \(\mathbb{R}^*\), the probability that \(\beta\) is exactly the solution to the constraint is bounded by the number of solutions divided by the size of \(\mathbb{R}^*\).

Given the number of solutions is at most 1, it follows that:
\begin{equation}    
    v \leq \frac{1}{|\mathbb{R}^{*}|}
\end{equation}
\end{proof}
\begin{theorem} \label{thm:verifiable}
    The $\ell$-server itED-PIR scheme $\Gamma$ in {\bf Fig.} \ref{fig:construction} constructed by itDPF shceme $\Pi$ is \(\left(t, \frac{1}{|\mathbb{R}^*|}\right)\)-verifiable if the scheme $\Pi$ is $t$-private, where \(|\mathbb{R}^*|  = p^\tau - p^{\tau-1}\).
\end{theorem}

\begin{proof}
    This theorem follows directly from Lemma \ref{Lemma:verifiability}.
\end{proof}

\begin{theorem}
    Suppose there exists an $\ell$-server $t$-private itDPF scheme $\Pi$ with output additive group $\mathbb{G} = (\mathbb{Z}_{p^\tau}, +)$ for database size $n \in \mathbb{N}$, which, on security parameter $\lambda \in \mathbb{N}$, outputs secret keys of length $L_{\mathsf{DPF}}$. Then, there is an $\ell$-server $t$-private itED-PIR scheme $\Gamma$ for binary databases $\mathbf{x} \in \{0,1\}^n$ with query complexity $\ell \cdot L_{\mathsf{DPF}}$ bits and answer complexity $\ell \cdot \tau \log p$ bits.
\end{theorem}

\begin{proof}
    The query complexity of $\Gamma$ is the total length of all queries $\{q_j\}_{j\in [\ell]}$. Each query $q_j$ is exactly a single itDPF key $k_j$ of length $L_{\mathsf{DPF}}$, so the total query complexity is $\ell \cdot L_{\mathsf{DPF}}$.

    The answer complexity is the total length of all answers $\{a_j\}_{j\in [\ell]}$. Each answer $a_j$ is an element of $\mathbb{G} = (\mathbb{Z}_{p^\tau}, +)$, which requires $\tau \log_2 p$ bits to represent. Thus, the total answer complexity is $\ell \cdot \tau \log_2 p$.
\end{proof}

By applying any of the $\ell$-server $t$-private itDPF schemes listed in Table~\ref{tab:existing_itdpf}, the following concrete instantiations are obtained:
\begin{corollary} \label{coro:1}
    Given a security parameter $\lambda \in \mathbb{N}$ and database size $n \in \mathbb{N}$, there exist a $3$-server statistical $1$-private $(1, \frac{1}{p-1})$-verifiable itED-PIR scheme for binary databases with the communication complexities 
    $O\big( \lambda \log p \cdot 2^{c(p)\sqrt{\log n \log \log n}}\big)$ (constructed from \cite{boyle2023information}, Theorem $9$)
    and a $4$-server statistical $1$-private $(1, \frac{1}{p-1})$-verifiable itED-PIR scheme for binary databases with the communication complexities
    $O\big(\lambda \cdot 2^{10\sqrt{\log n \log \log n}} + \lambda \log p\big)$ (constructed from \cite{li2025efficient}, Theorem $11$).

    Here, $p$ is a prime, with $c(2)=6$, $c(3)=10$, and $c(p)=2p$ for any $p \geq 5$. Both schemes provide statistical $1$-privacy.
\end{corollary}

\begin{corollary} \label{coro:2}
    Given database size $n \in \mathbb{N}$, there exists a $4$-server perfect $1$-private $(1, \frac{1}{p^\tau - p^{\tau-1}})$-verifiable itED-PIR scheme for binary databases with communication complexity
       $ O\big(\tau \log p \cdot 2^{c(p)\sqrt{\log n \log \log n}}\big)$ (constructed from \cite{boyle2023information} Theorem $5$ and \cite{li2025efficient} Theorem $6$).
    
    Here, $p$ is a prime, with $c(2)=6$ and $c(p)=2p$ for any $p \geq 3$. The scheme provides perfect $1$-privacy. 
\end{corollary}

\begin{corollary} \label{coro:3}
    Given database size $n \in \mathbb{N}$, there exists a $8$-server perfect $1$-private $(1, \frac{1}{p-1})$-verifiable itED-PIR scheme for binary databases with communication complexity
       $ O\big(2^{10\sqrt{\log n \log \log n}} +\log p\big)$ (constructed from \cite{li2025efficient} Theorem $8$).
    
    Here, $p$ is a prime. The scheme provides perfect $1$-privacy. 
\end{corollary}

\begin{corollary}  \label{coro:4}
    Given database size $n \in \mathbb{N}$, there exists a $d(t +1)$-server perfect $t$-private $(t, \frac{1}{p-1})$-verifiable itED-PIR scheme for binary databases with communication complexity
       $ O\big( \log p \cdot n^{1/ \lfloor (2d+1)/t \rfloor }\big)$ (constructed from \cite{li2025efficient} Theorem $7$).
    
    Here, $d,t$ are positive integers, $p$ is a prime. The scheme provides perfect $t$-privacy. 
\end{corollary}

\subsection{Discussion} \label{sec:discussion}

\paragraph{Advantages compared with itAPIR}
\begin{itemize}
    \item The itAPIR scheme employs $2$ itDPF keys in its query algorithm, whereas the itED-PIR scheme proposed in this work only uses 1 itDPF key. For certain itDPF schemes where the key size is asymptotic to a constant relative to the output group size (e.g., \cite{li2025efficient}, Theorem $8$), the itED-PIR scheme constructed upon such itDPF schemes halves the query communication complexity compared with the itAPIR scheme.

    \item The itAPIR scheme stores data over a finite field, while the itED-PIR scheme in this work is built over a ring. The advantage of a ring over a finite field is that there exist rings whose corresponding additive groups are $(\mathbb{Z}_{p^\tau},+)$ for $\tau \geq 2$, whereas no such finite fields exist. 
    
    Taking the construction of a 4-server perfect 1-private $(1, \epsilon)$-verifiable scheme as an example, Corollary \ref{coro:2} allows the selection of parameters $p=2$ and $\tau = -\log_2 \epsilon$, yielding an itED-PIR scheme with communication complexity $O\bigl(\tau \log p \cdot 2^{6\sqrt{\log n \log \log n}}\bigr)= O\bigl(\frac{1}{\epsilon} \cdot 2^{6\sqrt{\log n \log\log n}}\bigr)$. In contrast, the itAPIR scheme can only be constructed by choosing a prime $p > 1/\epsilon$ with $\tau = 1$, leading to a scheme with communication complexity $O\bigl(\tau \log p \cdot 2^{2p\sqrt{\log n \log\log n}}\bigr)$. 
    For the security level of $\epsilon=2^{-128}$ (a common choice for high-security applications), this translates to concrete parameter values: the itED-PIR scheme uses $p=2$ and $\tau=128$, while the itAPIR scheme requires a prime $p>2^{128}$ with $\tau=1$. The exponential term of the itAPIR scheme thus becomes $2^{2\cdot2^{128}\cdot\sqrt{\log n \log\log n}}$, whose growth is completely intractable for practical database sizes $n$.
    In practice, $\epsilon$ is typically set to $2^{-40}, 2^{-60}, 2^{-80},2^{-128}$ for different security levels, and the asymptotic communication complexity of the itAPIR scheme becomes prohibitive under these choices of $\epsilon$.
\end{itemize}

\paragraph{Resilience Against Cyberattacks}
The proposed itED-PIR scheme exhibits strong resilience against various cyberattacks, as elaborated below:
\begin{itemize}
    \item {\bf Eavesdropping Attacks} 
    All queries are transmitted in a privacy-preserving manner. The $t$-privacy guarantee ensures that eavesdroppers who intercept no more than $t$ queries cannot extract any meaningful information about the client’s query index.
    \item {\bf Collusion Attacks}  
    The $t$-private property of the scheme ensures that up to $t$ colluding servers cannot learn any non-negligible information about the client’s retrieval index.
    \item {\bf Man-in-the-Middle Attacks}
    The verifiable property enables the client to identify the presence of malicious servers—or more precisely, incorrect responses. Its $(t,\epsilon)$-verifiable guarantee ensures that the success probability of man-in-the-middle attacks, launched by adversaries who control no more than $t$ servers and tamper with their responses, is bounded by $\epsilon$.
\end{itemize}

As the proposed scheme achieves information-theoretic security, it offers distinct advantages over computationally secure alternatives. Unlike computationally secure schemes that rely on unproven computational hardness assumptions (e.g., lattice-based, number-theoretic, or symmetric cryptography assumptions), the security of the itED-PIR scheme does not depend on any limitations of an adversary’s computational power. This advantage makes the scheme naturally resilient to quantum threats, which could potentially break most computationally secure cryptosystems in the future.

{
\paragraph{Limitations}  
Despite the significant advantages, this study has several limitations that merit attention. 
First, the proposed scheme assumes an honest-but-curious client model, which may not cover all real-world adversarial scenarios where the client could be malicious. 
Second, this work focuses on the theoretical design and information-theoretic security analysis, and lacks a formal treatment of computationally secure ED-PIR schemes derived from computationally secure DPFs—a gap that also exists in the prior APIR work \citep{colombo2023authenticated}.

\paragraph{Future Work} 
Building on the limitations identified above, several future research directions can be explored. 
First, extend the security model to accommodate malicious clients, and formally analyze the scheme’s resilience against adversarial clients that may deviate from the protocol specification. 
Second, fill the theoretical gap by constructing and formally proving the security of computationally secure variants of the proposed ED-PIR scheme, which are derived from computationally secure DPFs.

Although practical implementation is not the core focus of this study, the following strategies are proposed for future practical deployment.
First, future work shall develop optimized implementation frameworks tailored to mainstream distributed storage systems (e.g., Hadoop, IPFS). The concrete instantiations presented in this work (Corollaries $1$–$4$) are designed for specific parameter settings, which may not directly satisfy all practical deployment requirements. Thus, adaptive parameter tuning and optimization are needed to balance security, efficiency, and resource constraints in real-world scenarios.
Second, in practical deployments, database sharding can be employed: by selecting an appropriate entry size to partition the database into smaller chunks, the number of entries per query can be reduced, which significantly improves both computational and communication efficiency.
}

\section{Conclusion}
This paper presents a novel $t$-private $\ell$-server information-theoretic error-detecting Private Information Retrieval (itED-PIR) scheme constructed over the ring $\mathbb{R}=(\mathbb{Z}_{p^\tau},+,\cdot)$ and instantiated from information-theoretic Distributed Point Functions (itDPFs), overcoming the finite field constraint of state-of-the-art information-theoretic authenticated Private Information Retrieval (itAPIR) schemes and enabling native compatibility with prime-power-order output groups $\mathbb{G}=(\mathbb{Z}_{p^\tau},+)$ for $\tau\geq2$. 
Distinguished from itAPIR's dual-key design, this work uses a single itDPF key per server, halving the query-side communication complexity.
The ring structure in this work achieves high efficiency than the finite field structure in itAPIR. 
Rooted in information-theoretic security, this work achieves inherent quantum resistance and robust defense against collusion, tampering, and eavesdropping attacks.

\section*{Acknowledgements}
The authors thank the anonymous reviewers for their helpful comments.


\begin{appendices}

\section{Generalized itED-PIR for Databases with $m$-bit Entries} \label{sec:mbit}
    \begin{figure*}[h] 
	\begin{boxedminipage}{\textwidth}
        \underline{Notations}
        \begin{itemize}
            \item $\ell$: The total number of servers.
            \item $t$: The number of servers that may collude to learn the retrieval index. 
            \item ${\bf x} = (x_1,\ldots,x_n)$: The binary database, which is a vector in $(\{0,1\}^{m})^n$.
            \item $\alpha$: The client's retrieval index.
            \item $\mathbb{R}$: The ring $\mathbb{R} = (\mathbb{Z}_{p^\tau},+,\cdot)$ (where $p$ is a prime and $\tau$ is a positive integer).
            \item $\mathbb{G}$: The additive group of the ring $\mathbb{R}$, $\mathbb{G} = (\mathbb{Z}_{p^{\tau}}, +)$.
            \item $\Pi = ({\sf Gen}, \{ {\sf Eval}_j\}_{j\in [\ell]})$: A $t$-private $\ell$-server itDPF scheme with output group $\mathbb{G}$.
        \end{itemize}
                
        \underline{$ ( \{q_j\}_{j\in [\ell]}, {\sf aux} ) \leftarrow {\sf Que}(1^{\lambda}, n,\alpha)$}
        \begin{itemize}
		\item[1.]
            Uniformly sample an element $\beta \in \mathbb{R}^*$. Generate DPF keys via $\{k_j\}_{j \in [\ell]} \leftarrow {\sf Gen}(1^\lambda, f_{\alpha, \beta})$, where $f_{\alpha, \beta}$ is the point function. For each $j \in [\ell]$, set query $q_j = k_j$.
            
            \item[2.] 
            Output queries $\{q_j\}_{j \in [\ell]}$ and auxiliary information ${\sf aux} = \beta$.

        \end{itemize}
        
        \underline{$ a_j \leftarrow {\sf Ans}({\bf x}, q_j,j)$}
        \begin{itemize}
                \item[1.]
                Evaluate the answer
                $
                a_j = \sum_{i \in [n]} x_{i} \cdot  {\sf Eval}_j (q_j,i),
                $
			and output $a_j$.
        \end{itemize}

        \underline{${\sf res} \leftarrow {\sf Rec}(\{a_j\}_{j \in [\ell]}, {\sf aux})$}      
        
        \begin{itemize}      
            \item[1.] 
            Compute the aggregate result	
            $
            R_{\text{agg}}=\sum_{j\in [\ell]} a_{j}.
            $

            \item[2.]
            Extract $\beta$ from ${\sf aux}$. If $\beta^{-1} \cdot R_{\text{agg}} \in \{0,1\}$, output $\beta^{-1} \cdot R_{\text{agg}}$, otherwise output $\perp$.
		\end{itemize}
	\end{boxedminipage}
	\caption{$\ell$-Server itED-PIR scheme $\Gamma'$ for $t$-private itDPF $\Pi$ with output group $\mathbb{G} =  (\mathbb{Z}_{p^{\tau}},+)$. This construction targets $m$-bit data entry for arbitrary $m$.}
	\label{fig:mbit construction}
\end{figure*}

The itED-PIR scheme $\Gamma$ presented in Fig. \ref{fig:construction} is directly applicable to databases with $m$-bit data entries for arbitrary positive integer $m$ without any modification. 
Specifically, the correctness and (perfect or statistical) $t$-privacy properties of the scheme remain fully preserved without any changes. Only the verifiability parameter is updated: the original $(t, \frac{1}{|\mathbb{R}^\ast|})$-verifiability (Theorem \ref{thm:verifiable}) is generalized to $(t, \frac{2^m - 1}{|\mathbb{R}^\ast|})$-verifiability for $m$-bit data entry databases. The generalized construction for $m$-bit data entries and its property proof are given as follows.

    \begin{theorem} \label{thm:correctness2}
        The itED-PIR scheme $\Gamma'$ is correct (Definition \ref{def:edpir-correctness}) if the itDPF scheme $\Pi$ is correct (Definition \ref{def:itdpf-correctness}).
    \end{theorem}
    \begin{proof}
        The proof is the same as the proof of Theorem \ref{thm:correctness}.
    \end{proof}

    \begin{theorem} \label{thm:privacy2}
        The itED-PIR scheme $\Gamma$ is perfect $t$-private (Definition \ref{def:itedpir-perfect-tprivacy}) if the underlying itDPF scheme $\Pi$ is perfect $t$-private (Definition \ref{def:itdpf-perfect-privacy}). Similarly, $\Gamma$ is statistical $t$-private (Definition \ref{def:itedpir-statistical-tprivacy}) if $\Pi$ is statistical $t$-private (Definition \ref{def:itdpf-statistical-privacy}).
    \end{theorem}
    \begin{proof}
        The proof is the same as the proof of Theorem \ref{thm:privacy}.
    \end{proof}

    \begin{lemma} \label{Lemma:verifiability2}
        Let \(\Gamma'\) be the \(\ell\)-server \(t\)-private itED-PIR scheme introduced in {\bf Fig.} \ref{fig:mbit construction}. For any database size \(n \in \mathbb{N}\), any database \({\bf x} = (x_1, \ldots, x_n) \in (\{0,1\}^m)^n\), any index \(\alpha \in [n]\), any subset \(T \subseteq [\ell]\) with \(|T| \leq t\), and any non-zero offsets \(\{\Delta_j\}_{j \in [\ell]}\) (where \(\Delta_j = 0\) for \(j \notin T\) and \(\Delta = \sum_{j \in [\ell]} \Delta_j \neq 0\)), the following bound holds:
        \begin{equation}
        \Pr\left[ 
        \begin{aligned}    
        y \in \{0,1\} \setminus \{x_\alpha\} :&  \\
        (\{q_j\}_{j \in [\ell]}, {\sf aux}) \leftarrow& {\sf Que}(1^\lambda, n, \alpha) \\
        a_j \leftarrow& {\sf Ans}({\bf x}, q_j, j) \quad \forall j \in [\ell] \\
        a'_j =& a_j + \Delta_j \quad \forall j \in [\ell] \\
        y \leftarrow& {\sf Rec}({\sf aux}, \{a'_j\}_{j \in [\ell]})
        \end{aligned}
        \right]
        \leq \frac{2^m - 1}{|\mathbb{R}^*|},
        \end{equation}
        where \(\mathbb{R} = (\mathbb{Z}_{p^\tau}, +, \cdot)\) is the underlying ring, \(\mathbb{R}^*\) denotes its multiplicative group (with order \(|\mathbb{R}^*| = p^\tau - p^{\tau-1}\) by Euler's totient function), and the probability is taken over all random coins used by the algorithms.
    \end{lemma}
        
\begin{proof}
Let \(\beta\) be the element uniformly sampled from \(\mathbb{R}^*\), \(R_{\text{agg}} = \sum_{j \in [\ell]} a_j\) denote the honest aggregate response, and \(\Delta = \sum_{j \in [\ell]} \Delta_j \neq 0\) denote the aggregate tampering offset. By the correctness of \(\Gamma\) (Theorem \ref{thm:correctness}), it holds that:
\begin{equation}
\beta^{-1} \cdot R_{\text{agg}} = x_\alpha \in \{0,1\}^m.
\end{equation}

A ``false acceptance" occurs if and only if the reconstructed result lies in the verification set \(\{0,1\}^m\) but is not the true value \(x_\alpha\), which is equivalent to:
\[
\beta^{-1} \cdot \left(R_{\text{agg}} + \Delta \right) = \hat{x}_\alpha \in \{0,1\}^m \setminus \{x_{\alpha}\}.
\]

Therefore, 
\begin{equation}
    \begin{aligned}
        v 
        &= \Pr\left[\bigcup_{\hat{x}_{\alpha} \in \{0,1\}^m \setminus x_{\alpha}} \Delta = \beta \cdot (\hat{x}_{\alpha} - x_{\alpha})    \right] \\
        & \leq \sum_{\hat{x}_{\alpha} \in \{0,1\}^m \setminus x_{\alpha}} \Pr\left[\Delta = \beta \cdot (\hat{x}_{\alpha} - x_{\alpha})\right] 
    \end{aligned}
\end{equation}

For the multiplicative group \(\mathbb{R}^*\), the non-zero linear constraint has at most one solution. Since \(\beta\) is uniformly sampled from \(\mathbb{R}^*\), the probability that \(\beta\) is exactly the solution to the constraint is bounded by the number of solutions divided by the size of \(\mathbb{R}^*\).

Given that the number of solutions is at most 1 for each term, it follows that:
\begin{equation}
    v \leq \frac{2^m-1}{|\mathbb{R}^{*}|}
\end{equation}
\end{proof}

\begin{theorem} \label{thm:verifiable2}
    The itED-PIR scheme $\Gamma$ is \(\left(t, \frac{2^m - 1}{|\mathbb{R}^*|}\right)\)-verifiable if the underlying itDPF scheme $\Pi$ is $t$-private, where \(|\mathbb{R}^*|  = p^\tau - p^{\tau-1}\).
\end{theorem}

\begin{proof}
    This theorem follows directly from Lemma \ref{Lemma:verifiability2}.
\end{proof}

\end{appendices}







\bibliography{manuscript}

\end{document}